\documentclass[a4paper, conference, 10pt] {IEEEtran}
\usepackage{amsmath,amssymb}
\usepackage{amsthm}
\usepackage{algcompatible}
\usepackage{algorithm}
\usepackage{algpseudocode}
\usepackage[pdftex]{graphicx}
\usepackage{tikz}
\usetikzlibrary{arrows}
\usepackage{pgfplots}
\usepackage[a4paper,top=1.9cm,inner=1.3cm,outer=1.3cm,bottom=2.9cm]{geometry}
\usepackage{multirow}
\usepackage{hhline}
\makeatletter
\newcommand*{\rom}[1]{\expandafter\@slowromancap\romannumeral #1@}
\makeatother
\ifCLASSINFOpdf
\else
\fi
\begin{document}
\title{A Distributed Algorithm for Quality-of-Service Provisioning in Multihop Networks}

\author{\IEEEauthorblockN{Ashok Krishnan K.S. and Vinod Sharma}
\IEEEauthorblockA{Dept. of ECE, Indian Institute of Science, Bangalore, India\\
%\\
%Indian Institute of Science\\
%Bangalore, India 560012\\
Email: \{ashok, vinod\}@ece.iisc.ernet.in
}}


%Telephone: (800) 555--1212, Fax: (888) 555--1212}


\maketitle

\begin{abstract}
We present a distributed algorithm for joint power control, routing and scheduling in multihop wireless networks. The algorithm also provides for Quality of Service (QoS) guarantees, namely, end-to-end mean delay guarantees and hard deadline guarantees, for different users. It is easily implementable and  works by giving local dynamic priority to flows requiring QoS, the priority being a function of the queue length at the nodes. We provide theoretical bounds for the stability properties of the algorithm. We also compare the performance of the algorithm with other existing algorithms by means of extensive simulations, and demonstrate its efficacy in providing QoS on demand.
\end{abstract}

% \fi
\IEEEpeerreviewmaketitle

\section{Introduction and Literature Review}
A wireless network consists of a number of nodes connected by time-varying channels, with stochastic arrival of data at various nodes, destined to other nodes. The network control problem consists of making decisions about the power control, routing and scheduling to be done at the different nodes. Wireless network control may be accomplished in a centralized or distributed manner. Centralized control is often difficult to implement, is memory intensive, time consuming and prone to failure (when the central controlling node fails). Distributed algorithms offer ease of implementation and more robustness to failure as opposed to centralized algorithms, though their performance may not be as good, given that distributed control uses only a subset of the total information about the system, as well as due to the fact that the action space is limited. In many real systems, implementing centralized control may not be feasible, given memory and time constraints. Developing distributed algorithms that can match the performance of centralized algorithms has always been a challenging problem\cite{Liu}.\\
\indent The notion of a \emph{throughput optimal} scheduling scheme was introduced in \cite{Tas92} for multihop wireless networks. A policy is \emph{throughput optimal} if it can stabilize an arrival process which can be stabilized by any policy. The \emph{capacity region} of the network  is defined to be the set of all mean arrival rate vectors to the network that can be stabilized by some policy. The notion of throughput-optimality was extended to include power control as well in \cite{Nee05}, where it was shown that the power allocation and scheduling that maximizes the sum of the rate-differential backlog product  is throughput optimal. It also provided a distributed algorithm version of this scheme, which, however, is not throughput optimal. In \cite{Sri50} the authors propose a distributed scheme that is guaranteed to achieve at least one-third of the capacity region, by generating a maximal matching between the nodes. However it does not study the most general SINR interference model; instead,  graph based interference models are used. In such models, one considers the interference graph of the network: links which interfere with each other cannot simultaneously transmit, whereas those which do not interfere can. A similar model is studied in \cite{Zuss}.  A scheme which maximizes the expected value of the  rate-differential backlog metric was proposed in \cite{Kim07}, which uses the SINR interference model.\\
\indent In \cite{Lee09} a distributed scheme for joint power control, scheduling and routing is proposed for wireless networks, that guarantees the attainment of a $\rho$ fraction of the capacity region $\Lambda$, under the SINR model. This is an extension of the scheduling policy described in\cite{Tas98}. In the distributed scheme in \cite{Tas98} one picks an activation scheme (or schedule) randomly, such that there is a nonzero probability of picking the optimal scheme; compares this choice with the previous slot's scheme in terms of its metric performance, and chooses the activation scheme which is better. It can be shown that this activation scheme converges to the optimal schedule almost surely. Gossip algorithms have been used in  \cite{Lee09} to calculate, in a distributed manner, the global metric required for making a scheduling decision. An extensive survey of various techniques and results in Gossip algorithms is available in \cite{ShaBook}.\\
\indent Quality-of-Service (QoS) requirements may encompass a wide range of requirements, such as stability, mean delay guarantees, delay deadlines, rate/ bandwidth guarantees and so on. The kind of QoS that a flow demands depends on the application that generates the flow. There have been a number of approaches to consider these QoS requirements separately as well as in combination. In the case of average delay constraints, one approach is to use the notions of \emph{effective bandwidth} and \emph{effective capacity} from large deviations theory, which lets us translate delay or queue length bounds into equivalent rate constraints, and solve that problem in the physical layer \cite{LauSurvey}. However, such schemes are accurate approximations to the actual requirments only when the queue lengths are large. Also, this method is not easy to apply in the multihop context, since the coupling between queues is complex, and translation of delay constraints to control actions is generally quite complex.\\
\indent In \cite{Chen99}, each node continuously keeps track of the minimum end-to-end delay, bandwidth and cost from that node to every other destination node. Given the QoS requirements for a flow, multiple paths are probed, from source to destination, and a feasible path is chosen using a scheme of forwarded \lq tickets\rq, which will collect the delay information along feasible paths. In \cite{BocheStan}, a one-to-one relationship is assumed between the given QoS constraint and the SINR. Thus, one can convert QoS constraints to SINR constraints. Under the additional assumption that the function mapping the feasible QoS set to the corresponding SINR values is log-convex, one can show that the feasible QoS region is a convex set. However, this additional assumption may not always hold. In \cite{Satya}, the authors study the problem of minimizing power while ensuring QoS in a network.\\
\indent Another traditional approach is to use Markov Decision Processes\cite{LauSurvey}. In multihop scenarios, however, such an approach is generally not tractable owing to the huge dimension of the state space. Lyapunov optimization based approaches\cite{GeorgBook} are generally considered more suitable in the multihop setting.\\
Our algorithm tries to provide QoS for flows that demand a mean delay guarantee, or a hard deadline, in a distributed fashion in a multihop network. Since we also ensure stability of all the queues, it does ensure that if a flow is coming at a certain rate, it gets that as its service rate along the way to the destination. These three QoS ensure that the main applications in the network: file transfers, real time applications-VoIP, teleconferencing and video streaming-will be satisfactorily served. Owing to the time-varying nature of the channels, we are considering a system where the paths are not fixed, which is a practical consideration. In such a scenario, one cannot decompose the QoS requirement to the level of paths and links. At this level of distributed decision making, providing QoS can be a challenging problem. Our algorithm attempts to combine simplified distributed decision making with the QoS problem.\\
\indent The rest of this paper is organized as follows: Section \rom{2} provides the system model and Section \rom{3} describes the distributed algorithm. In Section \rom{4} we obtain some theoretical results about the performance of the algorithm, and in Section \rom{5} we report some simulation results. 

\section{System Model}
\begin{figure}
	\centering
	\setlength{\unitlength}{1cm}
	\thicklines
	\begin{tikzpicture}[scale=0.8, transform shape]		
	\node[draw,shape=circle, fill=blue!30, scale=0.6, transform shape] (v1) at (4.7,0.5) {$node\ k$};
	\node[draw,shape=circle, fill=blue!30, scale=0.6, transform shape] (v2) at (2.4,1) {$node\ j$};
	\node[draw,shape=circle, fill=blue!30, scale=0.6, transform shape] (v3) at (3,3.5) {$node\ m$};
	\node[draw,shape=circle, fill=blue!30, scale=0.6, transform shape] (v4) at (0.3,5) {$node\ n$};
	\node[draw,shape=circle, fill=blue!30, scale=0.6, transform shape] (v5) at (5.1,5.1) {$node\ i$};
	\node[draw,shape=circle, fill=blue!30, scale=0.6, transform shape] (v6) at (7,2.2) {$node\ p$};
	\node[draw,shape=circle, fill=blue!30, scale=0.6, transform shape] (v10) at (0,2.5) {$node\ l$};
	\node (v7) at (7.6,5.25) {$A_i^j(t)$};
	\node (v8) at (6.5,5.8) {$q_i^j(t)$};
	\node (v9) at (3,5.4) {$\gamma_{ij}(t)$};			
	\draw[line width=1.0mm] (v2) -- (v1)
	(v4) -- (v5)			
	(v3) -- (v5)
	(v1) -- (v5)
	(v4) -- (v3)
	(v3) -- (v5)
	(v3) -- (v2)
	(v5) -- (v6)
	(v4) -- (v2)
	(v4) -- (v10)
	(v10) -- (v2)
	(v1) -- (v6);
	\draw[line width=0.5mm, line cap=round](5.8,5)--(6.7,5);
	\draw[line width=0.5mm, line cap=round](5.8,5.5)--(6.7,5.5);
	\draw[line width=0.5mm, line cap=round](5.8,5)--(5.8,5.5);
	\draw[line width=0.2mm](5.9,5)--(5.9,5.5);
	\draw[line width=0.2mm](6.0,5)--(6.0,5.5);
	\draw[line width=0.2mm](6.1,5)--(6.1,5.5);
	\draw[line width=0.2mm](6.2,5)--(6.2,5.5);
	\draw[line width=0.2mm](6.3,5)--(6.3,5.5);
	\draw[thick,->] (7.1,5.25) -- (6.6,5.25);
	\end{tikzpicture}
	\caption{A simplified depiction of a Wireless Network}
	\label{fig1}
\end{figure}
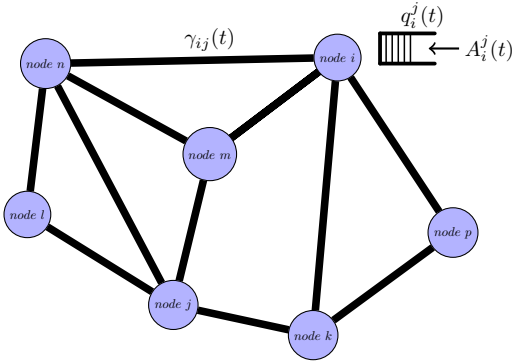

We consider a multihop network (Fig. 1), given by a graph $G=(V,E)$ where $V=\{1,2,..,N\}$ is the set of $N$ vertices and $E$, the set of links on $V$. We assume a slotted system, with the discrete time index denoted by $t\in \{0,1,2,...\}$.\\
\indent Each link $i-j$ has a time varying channel gain $\gamma_{ij}(t)$ at time $t$. The channel gain vector is represented by $\gamma(t)=[\gamma_{ij}(t)]_{1\leq i,j\leq N}$. At each node $i$, $A_i^c(t)$ denotes the i.i.d process of exogeneous arrival of packets destined to node $c$. The mean arrival rate is $\lambda_i^c=\mathbb{E}[A_i^c(t)]$, and the mean arrival rate vector is $\lambda = [\lambda_i^c]_{1 \leq i,c \leq N}$.  Any node $i$ for which there exists a node $c\neq i$ such that $\lambda_i^c(t) >0$ is called a source node, and  $c$ is the corresponding destination node. All traffic in the network with the same destination $c$ is called  \emph{flow} $c$; the set of all flows is denoted by $F$. The subset of flows which have QoS constraints is denoted by $F_Q$; each flow $c\in F_Q$ has an associated QoS criterion $\Theta(c)$. The criterion is a condition on some parameter of the flow, such as its delay. At each node there are queues, with $q_i^c(t)$ denoting the queue length at node $i$ storing packets corresponding to flow $c \in F$. \\
\indent By $p_{ij}(t)$ we denote the power at which node $i$ transmits to node $j$, at time $t$. The set of all transmit powers is denoted by the power vector, $p(t)=[p_{ij}(t)]_{1\leq i,j\leq N} $. The power vector takes its values from a feasible set $\mathcal{P}$, which is bounded.  The rate of transmission between node $i$ and node $j$ is $r_{ij}(t)=f(p(t),\gamma(t))$ where $f$ is some achievable rate function. We will be using the SINR rate function,
\begin{equation}
	r_{ij}(t)=\log_2\biggl(1+\dfrac{p_{ij}(t)\gamma_{ij}(t)}{\mathcal{N}_j(t)+\sum_{k\neq i}\sum_{l \in V} p_{kl}(t)\gamma_{kl}(t)}\biggr),
\end{equation}
\noindent with $\mathcal{N}_j(t)$ denoting the noise power at node $j$. This rate may be allocated to packets in one or more of the queues in node $i$, to be transferred to the corresponding queue in node $j$.\\
\indent The queue corresponding to flow $c(\neq i)$ at  node $i$ evolves as
\begin{equation}
	q_i^c(t)=q_i^c(t)-\mu_{OUT,i}^c(t)+\mu_{IN,i}^c+A_i^c(t),
\end{equation}
\noindent where $\mu_{OUT,i}^c(t)$ is the number of packets of flow $c$ that are routed out of node $i$ in slot $t$, and $\mu_{IN,i}^c(t)$ is the number of packets that are routed in. They obey the constraints
\begin{align*}
	\sum_c \mu_{OUT,i}^c(t)\leq \sum_j r_{ij}(t).
\end{align*}
We assume that, at a time, a node can be either a transmitter or a receiver, but not both, and it transmits to or receives from at most one node. Let us denote
%\vspace*{-5mm}
\begin{align*}
	\Delta_{ij} &=\max_c (q_i^c-q_j^c)^+,
\end{align*}
where $x^+=\max(x,0)$. Then, we define the maxweight policy as the policy that solves, at every time $t$,
\begin{align*}
	p^*(t)=\arg_{p\in\mathcal{P}}\max\sum_{ij}\Delta_{ij}r_{ij}(p),
\end{align*}
and schedules across each link $ij$ the flow $c$ for which $(q_i^c-q_j^c)^+=\Delta_{ij}$, and the corresponding link rate is $r_{ij}(p^*(t))$. We say that a queue $q_i^c(t)$ is stable under a policy if
\begin{align*}
	\lim_{T\to\infty}\sup\frac{1}{T}\sum_{\tau=0}^{T-1}\mathbb{E}[q_i^c(\tau)] &<\infty.
\end{align*}
The network is stable if all queues are stable.  The capacity region $\Lambda$ of the network is defined to be the convex hull of the set of all arrival rate vectors $\lambda$ that are stabilized by some policy\cite{Nee05}.
\section{A Distributed Scheme Providing QoS}
We propose a distributed algorithm (Algorithm 1) for joint scheduling, routing and power control, while also making provision for mean delay guarantees and hard deadline guarantees. This is an extension of the algorithm in \cite{Lee09}. However, it differs substantially from this algorithm on two counts: first, that it uses queue length information in the scheduling process, and second, that it makes provision for QoS as well. The use of queue length information is based on the intuitive idea of giving those nodes that have a higher queue length, a higher probability of becoming a transmitter. This should lead to improvement in performance. In this scheme, those links which have a high queue length at the transmitting side, and a low queue length at the receiving side, have a higher probability of being formed. This is a heuristic approach to backpressure.\\
\indent At the beginning of each slot, each node $i$ computes $q_i=\sum_{c\in F}h^c(q_i^c)$, where $h^c(x)=\theta x^2 \eta^c+x(1-\eta^c)$ if $c\in F_Q$ and $h^c(x)=x$ if $c\in F\setminus F_Q$, with $\theta>1$. Here, $\eta^c$ is one if the QoS constraint for flow $c$ was met in the previous time slot (i.e $\Theta(c)$ was satifsied), and is zero otherwise. This $q_i$ is a virtual queue length in the node, with extra weight being given to the backlogs of those flows whose QoS requirements were not met, thus capturing the dynamic priority given to flows requiring QoS. The nodes now use Algorithm 2, given below, to compute, in a distributed manner, $U^*$, which is a surrogate for $U=\sum_i u_i$, where $u_i=\min(q_i,B)$, with $B$ chosen to be a very large number. Node $i$ decides to be a transmitter with probability $\frac{u_i}{U^*}$; else, it becomes a receiver. As a result, nodes with a higher backlog of QoS packets have a higher probability of being a transmitter, and hence, pushing the packets out of itself. The queues with lesser backlog have a higher chance of being receivers. The algorithm thus dynamically moves packets from bigger queues to smaller queues.\\
\indent Each transmitter tries to randomly pair up with one of its neighbours, and establishes a link if the neighbour chosen was neither a transmitter nor paired with any other node. Each transmitter also picks a random power level for transmission. Over each link thus formed, we schedule the flow that maximizes  $(h^c(q_i^c)-h^c(q_j^c))^+$ if $\chi=1$. Else, we choose the flow that maximizes $(q_i^c-q_j^c)^+$. During the slots where $\chi=1$, this will prioritize flows to provide QoS. In other slots, this is needed for stability of the non-QoS flows. The variable $\chi$ captures the trade-off between stability of the system and QoS for some users.\\
\indent  We then compute the rate-differential backlog product over each link $ij$. Let $R_{ij}(t)=[r_{ij}(\tilde p(t))-(1-\alpha_2)r_{ij}(p(t-1))]$. The rate-backlog product, $M_{ij}$, is given by $R_{ij}(t)\Delta_{ij}$ if $\chi=0$, and $R_{ij}(t)(h^{c^*_{ij}}(q_i^{c^*_{ij}})-h^{c^*_{ij}}(q_i^{c^*_{ij}}))^+$ otherwise. We obtain $\tilde M$, an esimate of $\sum M_{ij}$ using Algorithm 2. If $\tilde M\geq 0$, we use the power $\tilde p_i$ at node $i$; else we use the power used in the previous slot, as well as the corresponding scheduling. To ensure that each node has knowledge of the rate at which it can transmit, all nodes are required to send out signals of $\nu \tilde p_i(t)$ and $\nu p_i(t-1)$ ($\nu$ being  sufficiently small) at the same time; as a result, each node may sense the power it receives, subtract the effect of its own power, and obtain its interference level without coming to know the entire channel state. This technique was used in \cite{Nee05}.\\
\indent The gossip algorithm we use works on the following principle: Say we have $K$ independent random variables distributed exponentially with parameters $y_1,y_2,..y_K$. Then the minimum of these random variables is an exponential random variable with parameter $y_1+y_2+..+y_K$. Hence, in order to compute the sum of $K$ values, exponential random variables with these values as parameters, and compute their minimum. The inverse of this random variable is an estimate for their sum. One may generate a number of such random variables and compute the corresponding inverse of their average, for more accuracy.\\
\indent The algorithm dynamically gives priority to the queues, depending on whether their QoS constraints have been met. The flows which fail to meet the QoS criterion are given higher weightage in the system, by means of the function $h$. The overall distributed algorithm is given as Algorithm 1, below.
\vspace*{-5mm}
\begin{algorithm}[H]
	\caption{Distributed Algorithm with provision for QoS}
	\label{DistriAlgo}
	\begin{algorithmic}[1]
		
		\If{$t=0$} $\eta^c \gets 0\  \forall c \in F$
		
		\EndIf
		
		\While{$t\geq 0$}

		\State Generate $\chi$, with $\mathbb{P}\{\chi=1\}=\sigma=1-\mathbb{P}\{\chi=0\}$. Communicate its value to all nodes by signaling.	
			
	%		\algstore{myalg}
	%	\end{algorithmic}
	%\end{algorithm}

	%\begin{algorithm} 	       
	%	\begin{algorithmic}[1]
	%		\algrestore{myalg}	
		
		\State At each node $i$	:

		\State Compute $q_i = \sum_{c\in F} h^c(q_i^c)$.

		\State  Generate  $\{X_i^j\}_{j=1}^L$, i.i.d exponential with parameter $u_i =\min(q_i,B)$. 
				
		\State By gossiping (Algorithm 2) estimate $X_{min}^j=\min_i \{X_i^j\}_{j=1}^L$.		
		
		\State Calculate  $U^*=\biggl(\dfrac{1}{L}\sum_{j=1}^L X_{min}^j\biggr)^{-1}$. 
		
		\State Generate $\phi\sim\mathcal{U}[0,1]$. 
		\If{$\phi<\dfrac{u_i}{U^*}$}$\ i \gets \textit{transmitter}$
		\Else $\ i \gets \textit{receiver}$ 
		\EndIf
		
		\State Each transmitter $i$ picks a power $p_i\sim\mathcal{U}[0,p_{max}]$. Pick a neighbour uniformly randomly and send a request to pair (RTP).
				
		\State Each receiver $j$ waits for an RTP, pairs up with the first transmitter that sends it an RTP.				
				
		\State Over any link $(i,j)$ formed, schedule $c^*_{ij}=\arg_{c\in F}\max \chi(h^c(q_i^c)-h^c(q_j^c))^++(1-\chi)(q_i^c-q_j^c)^+$.		
		
		\State Each paired transmitter $i$ beams $\nu \tilde p^i$ and $\nu p^i(t-1)$.
		
		\If{$\chi=0$} $\ M_{ij}\gets \Delta_{ij}R_{ij}(t)$ 	
		\Else $\ M_{ij}\gets (h^{c^*_{ij}}(q_i^{c^*_{ij}})-h^{c^*_{ij}}(q_i^{c^*_{ij}}))^+R_{ij}(t)$
		\EndIf	
		
		\State  Generate $\{Y_i^j\}_{j=1}^L$, i.i.d exponential with parameter $M_{ij}$.
		
		\State By gossiping (Algorithm 2) estimate $Y_{min}^j=\min_i \{Y_i^j\}_{j=1}^L$.		
		
		 \State Calculate  $\tilde M=\biggl(\dfrac{1}{L}\sum_{j=1}^L Y_{min}^j\biggr)^{-1}$. 		
		
		\State If $\tilde M\geq0$, use the power and scheduling generated in the current slot. Else, use the power allocation and scheduling from the previous slot.
		
		\State For each flow $c$:
		\If{QoS criterion $\Theta(c)$ was satisfied} $\eta^c\gets 0$
		\Else $\ \eta^c\gets 1$
		\EndIf
		\State Update this information in the network using gossiping.
		
		\EndWhile
		
	\end{algorithmic}
	
\end{algorithm}
\vspace*{-6mm}

\begin{algorithm}[h]
	\begin{algorithmic}[1]
		\State Each node $i$ has $L$ numbers $Z_i^1,\dots Z_i^L$ with parameter $z_i$.
		\While{$k=0,1,..,T$} at each node
		\State Choose a neighbour with probability $1/N$. Call it $j$. 
		\State $Z_i^l,Z_j^l\gets \min(Z_i^l,Z_j^l)$ for each $l=1,\dots, L$.
		\EndWhile
		
	\end{algorithmic}
	\caption{Gossip Algorithm}
	\label{GossiAlgo}
\end{algorithm}
\vspace*{-6mm}
\section{Performance Analysis}

Let $r_{ij}(p)$ denote the rate across link $ij$ under power allocation $p$. Denote the optimal rates in slot $t$ by $r_{ij}^*(p^*(t))$. Then we have the following Lemma, which is a version of Theorem 1 in \cite{Lee09}:
\vspace*{-5mm}
\newtheorem{lem}{Lemma}
\begin{lem}
	Let an algorithm have power allocation $p(t)$ and let the rate under its scheduling in time $t$ be $r_{ij}(p(t))$, for every link $ij\in E$. If there exist $\alpha_1,\alpha_2,\beta_1,\beta_2\in(0,1)$ such that for all $t$,
	\begin{align*}
		\mathbb{P}[\sum_{ij\in E}\Delta_{ij}r_{ij}(p(t)) &\geq(1-\alpha_1)\sum_{ij\in E}\Delta_{ij}r_{ij}^*(p^*(t))]\geq\beta_1,\\
		\mathbb{P}[\sum_{ij\in E}\Delta_{ij}r_{ij}(p(t)) &\geq(1-\alpha_2)\sum_{ij\in E}\Delta_{ij}r_{ij}(p(t-1))]\\
		&\geq1-\beta_2.
	\end{align*}
	
	Then, the algorithm will stabilize the network for any arrival rate vector $\lambda \in \rho\Lambda$ where $\rho<1-(\alpha_1+(1-\alpha_1)\alpha_2)-2\sqrt{\dfrac{\beta_2}{\beta_1}}$.
\end{lem}
\begin{proof} See \cite{Lee09}.
\end{proof}
%\vspace*{-5mm}
While $\rho$ may be a small number, the utility of this result lies in the fact that we can provide a stability result under very general rate models, including the SINR model, which is in general quite difficult to analyze. To show that our algorithm satisfies this theorem, we will need another result from \cite{ShaBook}:
\begin{lem} 
	\label{Lemmatwo}
	Let $\epsilon,\delta\in(0,\frac{1}{2})$. Let $L=3\delta^{-2}\log(4\epsilon^{-1})$. Assuming the gossiping matrix is complete, the gossiping algorithm computes an estimate $\tilde S$ of the sum $S$, with $\tilde S\in[(1-\delta)S,(1+\delta)S]$ for all nodes with probability greater than or equal to $1-\epsilon$ in time $T=O(\delta^{-2}\log N\epsilon^{-1}\delta^{-1})$.
\end{lem}
\begin{proof} See \cite{ShaBook}.
\end{proof}
\vspace*{-3.5mm}
\begin{lem}
	\label{Lemmathree}
	Let $\alpha_1\in(0,1)$. Then, Algorithm 1 produces rates $r_{ij}(p(t))$, which satisfy, at every time $t$,
	\begin{align*}
		\mathbb{P}[\sum_{ij\in E}\Delta_{ij}r_{ij}(p(t)) &\geq(1-\alpha_1)\sum_{ij\in E}\Delta_{ij}r_{ij}^*(p^*(t))]\geq\beta_1
	\end{align*}
	where $\beta_1=(1-\beta_3){\bigl(\frac{\epsilon}{2(1-\alpha_3^2)N^{3.5}{B^2}}\bigr)}^N$, with $\beta_3\in(0,1)$, $\alpha_3\in\bigl(0,\frac{1}{2NB}\bigr)$ and $\epsilon>0$.
\end{lem}
\begin{proof}
	In every slot, the probability of a node being a transmitter is $u_i/U^*$, where $u_i=\min(q_i,B)$ and $U^*$ is the estimate of $U=\sum_{j \in V} u_j$ obtained by gossiping. Since each $u_i$ is less than or equal to $B$, their sum, $U$, cannot exceed $NB$.\\	
	\indent Pick $\alpha_3\in\bigl(0,\frac{1}{2NB}\bigr)$, and  $\beta_3\in(0,1)$. It follows from Lemma \ref{Lemmatwo} that using Algorithm 2 for gossiping, and running for $O(\log(n{\beta_3}^{-1}{\alpha_3}^{-1})/{\alpha_3}^2)$ iterations, returns a value $U^*\in[(1-\alpha_3)U,(1+\alpha_3)U]$ with probability greater than or equal to $1-\beta_3$. We assume that the gossiping algorithm runs for this sufficient number of iterations. Conditioned on this event(which we will call $\mathcal{G}$), we have the probability of selecting any link $ab$, independent of other links, given by	
	\begin{align*}
		\mathbb{P}(\textrm{link $ab|\mathcal{G}$}) & \geq \mathbb{P}(\textrm{$a$ is txr$|\mathcal{G}$})\dfrac{\mathbb{P}(\textrm{$b$ is a rxr $|\mathcal{G}$})}{\textrm{(no: of neighbours of $a$)}}\\
		& \geq\dfrac{u_a}{U^*}\dfrac{1}{N}\biggl(1-\dfrac{u_b}{U^*}\biggr).
	\end{align*}
	Since $\mathcal{G}$ implies that $(1-\alpha_3)U\leq U^*\leq(1+\alpha_3)U$, we have
	\begin{align*}
		\mathbb{P}(\textrm{link $ab|\mathcal{G}$})  & \geq \dfrac{u_a}{NU(1+\alpha_3)}\biggl(1-\dfrac{u_b}{(1-\alpha_3)U}\biggr)\\
		&=\dfrac{u_a}{NU(1-\alpha_3^2)}\biggl(\dfrac{U-u_b}{U}-\alpha_3\biggr).
	\end{align*}	
	Since $U-u_b=\sum_{j \in V, j \neq b} u_j \geq u_a$, and $U \leq NB$,we have:	
	\begin{align*}
		\mathbb{P}(\textrm{link $ab|\mathcal{G}$})  & \geq \dfrac{u_a}{N^2B(1-\alpha_3^2)}\biggl(\dfrac{u_a}{NB}-\alpha_3\biggr)\\
		& \geq \dfrac{1}{N^2B(1-\alpha_3^2)}\biggl(\dfrac{1}{NB}-\alpha_3\biggr),
	\end{align*}	
	\noindent where we have assumed $q_a \geq 1$, since any node having total queue length equal to zero can be removed from the set of transmitters, without affecting the system's performance. 
	With $B$ being a large positive integer, $u_a=\min(q_a,B) \geq 1$.  Since we have chosen $\alpha_3\in(0,\dfrac{1}{2NB})$, we find that:	
	\begin{align*}
		\mathbb{P}(\textrm{link $ab|\mathcal{G}$})  \geq \dfrac{1}{2(1-\alpha_3^2){N^3}{B^2}}.
	\end{align*}
	Since the number of transmitter-receiver pairs (links) possible under our assumptions is less than $N$, the probability of choosing any particular configuration of links is bounded from below by $\bigl(\frac{1}{2(1-\alpha_3^2){N^3}{B^2}}\bigr)^N$. In particular, the probability of chosing the optimal link configuration is bounded below by this value. Since the power vector is chosen independent of the links, and is chosen uniformly randomly over the range $[0,p_{max}]^N$, the probability that the power vector is in an $\epsilon$ radius around the optimal power vector is bounded below by $(\frac{\epsilon}{N^{0.5}})^N$, assuming $p_{max}=1$ (See Lemma 4 of \cite{Lee09} for details).\\
	\indent Define the event
	\begin{align*}
	\mathcal{A}:=\{\sum_{ij\in E}\Delta_{ij}r_{ij}(p(t))\geq(1-\alpha_1)\sum_{ij\in E}\Delta_{ij}r_{ij}^*(p^*(t))\}.
	\end{align*}	
	Since $\sum_{ij\in E}\Delta_{ij}r_{ij}(p(t))$ is a continuous function of $p(t)$ for a fixed link configuration, for any $\alpha_1 \in (0,1)$, there exists $\epsilon>0$ such that $\mathcal{A}$ is true for any $p(t)$ which satisfies the event $\{||p(t)-p^*(t)||<\epsilon\}$. 
	We have
	\begin{align*}
		\mathbb{P}[\mathcal{A}|\mathcal{G}]	&\geq \mathbb{P}[\mathcal{A}|\mathcal{G},\mathcal{S}^*]\mathbb{P}[\mathcal{S}^*|\mathcal{G}] \\
		&\geq \mathbb{P}[\{||p(t)-p^*(t)||<\epsilon\}|\mathcal{G},\mathcal{S}^*]\mathbb{P}[\mathcal{S}^*|\mathcal{G}] \\
		&\geq \biggl(\frac{\epsilon}{N^{0.5}}\biggl)^N \biggl(\dfrac{1}{2(1-\alpha_3^2){N^3}{B^2}}\biggr)^N, 
	\end{align*}
	where $\mathcal{S}^*$ is the event corresponding to choosing the optimal link configuration.
	We know that
	\begin{align*}
		\mathbb{P}[\mathcal{A}]\geq \mathbb{P}[\mathcal{A}|\mathcal{G}]\mathbb{P}[\mathcal{G}],
	\end{align*}
	and since $\mathbb{P}[\mathcal{G}]=1-\beta_3$,the result follows. 
\end{proof}
\begin{lem}
	\label{Lemmafour}
	Let $\alpha_2,\beta\in(0,1)$. Then, Algorithm 1 produces rates $r_{ij}(p(t))$, which satisfy, at every time $t$,
	\begin{align*}
		\mathbb{P}[\sum_{ij\in E}\Delta_{ij}r_{ij}(p(t)) &\geq(1-\alpha_2)\sum_{ij\in E}\Delta_{ij}r_{ij}(p(t-1))]\\
		&\geq1-\beta_2,
	\end{align*}
	where $\beta_2=\beta+\sigma(1-\beta)$.
\end{lem}
\begin{proof}
	Let $\mathcal{H}$ be the event $\{\chi=0\}$. Conditioned on $\mathcal{H}$, at each transmitter $i$, we generate $L$ exponential random variables, with parameter equal to $M_{ij}=\Delta_{ij}[r_{ij}(\tilde p(t))-(1-\alpha_2)r_{ij}(p(t-1))]$. We need to estimate the sum $M=\sum_{ij}M_{ij}$, and if $M\geq0$, we go with the power allocation $\tilde p(t)$, else we use $p(t-1)$.\\	
	\indent Let $\alpha\in(0,1)$, and pick $L=3(\alpha)^{-2}\ln (4/\beta)$. Then, assuming the Gossiping Algorithm runs for $T=O(\log(N{\beta}^{-1}{\alpha}^{-1})/{\alpha}^2)$ iterations, it follows from Lemma \ref{Lemmatwo} that the estimate $\tilde M\in[(1-\alpha)M,(1+\alpha)M]$ with probability greater than or equal to $1-\beta$. Once these many iterations are complete, we have $\{M\geq0\}\iff\{ \tilde M\geq 0\}$.
	Define the event
	\begin{align*}
		\mathcal{B}:=\{\sum_{ij\in E}\Delta_{ij}r_{ij}(p(t)) &\geq(1-\alpha_2)\sum_{ij\in E}\Delta_{ij}r_{ij}(p(t-1))\}
	\end{align*}
	We can see that 
	\begin{align*}
		\mathbb{P}[\mathcal{B}|\mathcal{H}] &=\mathbb{P}[M\geq 0]=\mathbb{P}[\tilde M\geq 0]\geq (1-\beta).
	\end{align*}	
	Since $\mathbb{P}[\mathcal{B}]\geq \mathbb{P}[\mathcal{B}|\mathcal{H}]\mathbb{P}[\mathcal{H}] $ and $\mathbb{P}[\mathcal{H}]=1-\sigma$,the result follows.
\end{proof}
\newtheorem{thm}{Theorem}
\begin{thm}
	Algorithm 1 stabilizes the network for any arrival rate vector $\lambda \in \rho\Lambda$ where $\rho<1-(\alpha_1+(1-\alpha_1)\alpha_2)-2\sqrt{\dfrac{\beta_2}{\beta_1}}$.
\end{thm}
\begin{proof}
	Follows from Lemmas \ref{Lemmatwo}, \ref{Lemmathree} and \ref{Lemmafour}.
\end{proof}
Hence, we are guaranteed stability for all arrival rates in the region $\rho\Lambda$. Observe that the value of $\delta_1$ is decreasing as $B$ increases, the guarantee that one can give in terms of achievable capacity region decreases as a consequence. However, in simulations below we will see that increasing $B$, or letting it go to infinity, does not seriously hamper the performance of the algorithm in terms of its stability region. We also note that the value $\sigma$ captures a trade-off between QoS and stability.\\
\indent Comparing our algorithm with \cite{Lee09}, we can see that for the same values of $\alpha_1$ and $\alpha_2$, we can obtain a better $\rho$ by choosing corresponding values of $\sigma$ and $\beta$. This is borne out by the simulations where we compare the performance of the algorithms in terms of stability region. Also, via extensive simulations we have seen that the algorithm actually provides a much larger stability region than what is dictated by $\rho$. Thus, it is in fact a practically useful distributed algorithm which provides end-to-end QoS in a multihop wireless network.\\
\indent Even if the Gossip matrix is not complete, one may obtain the same result. However the number of timeslots in which one needs to operate the gossip algorithm will be much higher. Exact expressions may be calculated for these as well\cite{ShaBook}. \\
\indent One may observe that since the algorithm guarantees stability for all arrival rate vectors contained in $\rho\Lambda$, it naturally provides for rate guarantees for any flow that generates packets at a constant rate within this region.
\section {Simulation Results}
For the simulations, we  consider networks of 10, 15 and 20 nodes, with the nodes distributed randomly uniformly in a unit square. We  assume Rayleigh fading between the nodes, as well as that packet arrivals are i.i.d across slots with Poisson distribution. The rate function, as mentioned earlier, will be the SINR rate function. For all the simulations we will use $\sigma =0.999$ and $B =10^{5}$. While these values reduce the theoretical value of $\rho$, it is evident from the simulations that they enhance the performance.\\
\indent We first compare the stability region that our algorithm offers, and compare it to two distributed algorithms: Lee\cite{Lee09} and Distributed DRPC\cite{Nee05}. For a network of 20 nodes we see that our algorithm outperforms both the others in terms of stability, when the number of flows is five (Fig. \ref{fig2}), as well as when it is fifteen (Fig. \ref{fig3}). We plot the change in total queue length as arrival rate at all nodes is increased uniformly. From the figures it is clear that our algorithm offers a huge improvement as far as stability is concerned.
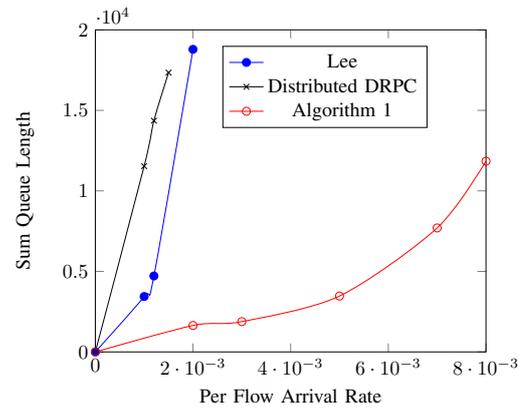
\begin{figure}
	\centering
	\begin{tikzpicture}[scale=0.75, transform shape]		
	\begin{axis}[
	xlabel=Per Flow Arrival Rate,
	ylabel=Sum Queue Length,
	xmax=0.008,
	xmin=0,
	ymax=20000,
	ymin=0,
	legend style = { at = {(0.85,0.95)}}
	]
	\addplot[smooth,mark=*,blue] plot coordinates {
		(0,0)
		(0.001,3439)
		(0.0012,4716)
		(0.002,18789)
	};
	\addlegendentry{Lee}
	
	\addplot[smooth,color=black,mark=x]
	plot coordinates {
		(0,0)
		(0.001,11537)
		(0.0012,14359)
		(0.0015,17344)
	};
	\addlegendentry{Distributed DRPC}
	\addplot[smooth,color=red,mark=o]
	plot coordinates {
		(0,0)
		(0.002,1643)
		(0.003,1881)
		(0.005,3469)
		(0.007,7698)
		(0.008,11844)
	};
	\addlegendentry{Algorithm 1}
	\end{axis}
	\end{tikzpicture}
	\caption{Stability Region for our algorithm for a network with 20 nodes and 5 flows}
	\label{fig2}
\end{figure}
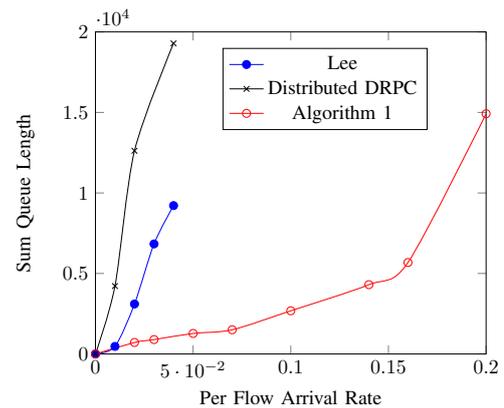
\begin{figure}
	\centering
	\begin{tikzpicture}[scale=0.75, transform shape]		
	\begin{axis}[
	xlabel=Per Flow Arrival Rate,
	ylabel=Sum Queue Length,
	xmax=0.2,
	xmin=0,
	ymax=20000,
	ymin=0,
	legend style = { at = {(0.85,0.95)}}
	]
	\addplot[smooth,mark=*,blue] plot coordinates {
		(0,0)
		(0.01,470)
		(0.02,3104)
		(0.03,6836)
		(0.04,9215)
	};
	\addlegendentry{Lee}
	
	\addplot[smooth,color=black,mark=x]
	plot coordinates {
		(0,0)
		(0.01,4219)
		(0.02,12617)
		(0.04,19289)
	};
	\addlegendentry{Distributed DRPC}
	\addplot[smooth,color=red,mark=o]
	plot coordinates {
		(0,0)
		(0.02,715)
		(0.03,895)
		(0.05,1275)
		(0.07,1503)
		(0.1,2681)
		(0.14,4305)
		(0.16,5682)
		(0.2,14918)
	};
	\addlegendentry{Algorithm 1}
	\end{axis}
	\end{tikzpicture}	
	\caption{Stability Region for our algorithm for a network with 20 nodes and 15 flows}
	\label{fig3}
\end{figure}\\
\indent The first QoS parameter that we will consider is mean delay guarantee. For such a flow $c$, at its destination node, the mean end-to-end delay is computed empirically, by averaging over all packets of that flow that arrive at the destination. If this value is greater than the mean delay required by the flow, the corresponding $\eta^c(t)$ is set to 1.  We present case studies of networks of 10 and 15 nodes, with the number of QoS flows being one or two. Each scenario is studied for a fixed value of the arrival rate vector, which is chosen within the capacity region of the network.\\
\indent Table \ref{table:One} gives the mean delay values of the QoS flow for two cases. Network 1 is a case of 10 nodes with 7 flows, of which one flow requires a mean delay guarantee. Network 2 is a case of 15 nodes with 10 flows of which one requires a mean delay guarantee. The value of the parameter  $\theta$ used for giving priority, is 10 in both cases. %Table \ref{table:Two} corresponds to a network of 10 nodes with 7 flows, of which two flows are mean delay constrained flows, and $\theta=5$. 
Table \ref{table:Three} is for 15 nodes with 7 flows, of which two flows require mean delay guarantees, and  $\theta=10$.\\
\indent From the simulations it is evident that the value of $\theta$ may be increased in order to gain a better performance. Also, in the case of multiple flows with QoS requirements, the flows are likely compete with each other as well, in order to have their share of the system resources. In Table \ref{table:Three}, both QoS flows are given the same priority (as indicated by $\theta$), one may also use different $\theta$ values corresponding to different flows. Due to the fact that the system is controlled in a distributed fashion, the number of  QoS demands it can support simultaneously may not be huge. One also observes that the mean delay cannot be brought down below a particular value. This in some sense is the limit of what the algorithm can achieve, for this particular form of the function $h$. This value is a function of the arrival rate vector.\\
\indent The next QoS parameter is hard deadline guarantee. In this case the QoS is specified by two values, a delay deadline $d$ and a dropping ratio $r$, and it is required that no more than $r$ fraction of the packets have a delay more than $d$. The value of $r$ is estimated empirically, and if this is greater than the required dropping ratio, the corresponding $\eta^c(t)$ is set to 1.\\
\indent Table \ref{table:Four} gives the delay performance of a 10 node network with 10 nodes and 8 flows, of which three are QoS flows: two have a mean delay requirement, and one has a hard deadline. To meet the hard deadline, the stability region has reduced. The hard deadline flow has to meet a delay deadline of 70. The mean delay flows have $h(x)=10x^2$ and the hard deadline flow has $h(x)=20x^2$. Note that the hard deadline is achieved for 94.9\%, 97\% and 98\% of the packets, as required, with little impact on the mean delay performance. Note that running the algorithm of \cite{Lee09} results in a mean delay of 127 and 104 respectively, for flow 1 and 2 respectively; and the drop ratio for flow 3 is $52.7\%$ (this is the fraction of packets that violates the end to end hard deadline). We see from simulations that we need to set the $\theta$ value for flows having hard deadline to be at least twice that for mean delay constrained flows.
%\vspace*{-5mm}
\begin{table}[h]
	\centering
	\caption{One flow with  mean delay requirement}
	\label{table:One}
	\begin{tabular}{|p{1.2cm}|p{1.2cm}|p{1.2cm}|p{1.2cm}|}
		\hline 
		\multicolumn{2}{|c}{Network 1}\ & \multicolumn{2}{|c|}{Network 2}\\
		\hline
		Delay Target (slots) & Delay Achieved (slots) & Delay Target (slots) & Delay Achieved (slots)\\
		\hline
		200 & 202 & 350 & 353\\
		\hline
		180 & 181 & 300 & 292\\
		\hline
		150 & 152 & 230 & 236\\
		\hline
		120 & 121 & 200 & 212\\
		\hline
		100 & 100 & 180 & 193\\
		\hline
		80 & 83 & 150 & 160\\
		\hline
		60 & 61 & 120 & 149\\
		\hline
	\end{tabular}	
\end{table}
%\begin{table}[h]
%	\begin{tabular}{|p{1.2cm}|p{1.2cm}|p{1.2cm}|p{1.2cm}|}
%		\hline 
%		\multicolumn{2}{|c}{Flow 1}\ & \multicolumn{2}{|c|}{Flow 2}\\
%		\hline
%		Delay Target (slots) & Delay Achieved (slots) & Delay Target (slots) & Delay Achieved (slots)\\
%		\hline
%		230 & 233 &230 & 231\\
%		\hline
%		200 & 210 & 200 & 199\\
%		\hline
%		200 & 198 & 160 & 165 \\
%		\hline
%		160 & 160 & 200 & 201\\
%		\hline
%		160 & 160 & 140 & 151\\
%		\hline
%		140 & 141 & 140 & 143\\
%		\hline
%		120 & 135 & 140 & 157\\
%		\hline
%	\end{tabular}
%	\caption{Two flows with mean delay requirement}
%	\label{table:Two}
%\end{table}

\vspace*{-4mm}
\begin{table}[h]
	\centering
		\caption{Two flows with mean delay requirement}
		\label{table:Three}
	\begin{tabular}{|p{1.2cm}|p{1.2cm}|p{1.2cm}|p{1.2cm}|}
		\hline 
		\multicolumn{2}{|c}{Flow 1}\ & \multicolumn{2}{|c|}{Flow 2}\\
		\hline
		Delay Target (slots) & Delay Achieved (slots) & Delay Target (slots) & Delay Achieved (slots)\\
		\hline
		300 & 308 & 300 & 330\\
		\hline
		250 & 248 & 250 & 256\\
		\hline
		200 & 210 & 250 & 270\\
		\hline
		150 & 169 & 200 & 202\\
		\hline
		180 & 182 & 180 & 189\\
		\hline
		160 & 185  &160 & 179\\
		\hline
	\end{tabular}
\end{table}
%\vspace*{-5mm}
\begin{table}[h]
	\centering
		\caption{Two mean Delays and one hard deadline}
		\label{table:Four}
	\begin{tabular}{|p{0.6cm}|p{0.6cm}|p{0.5cm}|p{0.6cm}|p{0.6cm}|p{0.5cm}|p{0.6cm}|p{0.6cm}|p{0.5cm}|}
		\hline
		\multicolumn{3}{|c}{Flow 1}\ & \multicolumn{3}{|c|}{Flow 2} & \multicolumn{3}{ c|}{Flow 3}\\		
		\hline
		Delay Target (slots) & Delay Achi- eved (slots) & Mean Delay in \cite{Lee09} (slots) & Delay Target (slots) & Delay Achi- eved (slots) & Mean Delay in \cite{Lee09} (slots) & Drop ratio Target &  Drop ratio Achi- eved & Drop Ratio in \cite{Lee09}\\
		\hline
		30 & 31 & \multicolumn{1}{c|}{\multirow{2}{*}{ }} & 40 & 41 & \multicolumn{1}{c|}{\multirow{2}{*}{}} & 5\% & 5.1\% & \multicolumn{1}{c|}{\multirow{2}{*}{}}\\
		\hhline{--~--~--~}
		30 & 31 & 127   & 40 & 41 & 104  & 3\% & 3\% & 52.7\% \\
		\hhline{--~--~--~}
		30 & 31 &   & 40 & 40 &   & 2\% & 2\% &  \\
		\hline
	\end{tabular}
\end{table}
\vspace*{-2mm}
\section{Conclusions and Future Work}
Distributed control of wireless networks with QoS is a challenging problem. We have obtained a distributed algorithm for routing, power control and scheduling of links using  queue length dependent cross-layer schemes under the SINR model of interference, while simultaneously providing  mean delay guarantees and hard deadline guarantees. Simulations demonstrate the improvement the scheme provides over existing approaches, as well as its ability to provide delays close to what is demanded by the users. The stability region expressions, as well as simulations point to how asking for more QoS effectively diminishes the amount of traffic the system can support. Theoretical bounds on how much QoS can be asked of a system is a question that can be explored in the future. Another question is the characterization of the upper limits of such distributed algorithms vis-\`a-vis centralized algorithms.

\bibliography{survey}

\begin{thebibliography}{10}
\providecommand{\url}[1]{#1}
\csname url@samestyle\endcsname
\providecommand{\newblock}{\relax}
\providecommand{\bibinfo}[2]{#2}
\providecommand{\BIBentrySTDinterwordspacing}{\spaceskip=0pt\relax}
\providecommand{\BIBentryALTinterwordstretchfactor}{4}
\providecommand{\BIBentryALTinterwordspacing}{\spaceskip=\fontdimen2\font plus
\BIBentryALTinterwordstretchfactor\fontdimen3\font minus
  \fontdimen4\font\relax}
\providecommand{\BIBforeignlanguage}[2]{{%
\expandafter\ifx\csname l@#1\endcsname\relax
\typeout{** WARNING: IEEEtran.bst: No hyphenation pattern has been}%
\typeout{** loaded for the language `#1'. Using the pattern for}%
\typeout{** the default language instead.}%
\else
\language=\csname l@#1\endcsname
\fi
#2}}
\providecommand{\BIBdecl}{\relax}
\BIBdecl

\bibitem{Liu}
X.~Liu and A.~Goldsmith, ``Wireless network design for distributed control,''
  in \emph{Decision and Control, 2004. CDC. 43rd IEEE Conference on},
  vol.~3.\hskip 1em plus 0.5em minus 0.4em\relax IEEE, 2004, pp. 2823--2829.

\bibitem{Tas92}
L.~Tassiulas and A.~Ephremides, ``Stability properties of constrained queueing
  systems and scheduling policies for maximum throughput in multihop radio
  networks,'' \emph{IEEE transactions on automatic control}, vol.~37, no.~12,
  pp. 1936--1948, 1992.

\bibitem{Nee05}
M.~J. Neely, E.~Modiano, and C.~E. Rohrs, ``Dynamic power allocation and
  routing for time-varying wireless networks,'' \emph{IEEE Journal on Selected
  Areas in Communications}, vol.~23, no.~1, pp. 89--103, 2005.

\bibitem{Sri50}
L.~Bui, A.~Eryilmaz, R.~Srikant, and X.~Wu, ``Joint asynchronous congestion
  control and distributed scheduling for multi-hop wireless networks.'' in
  \emph{INFOCOM}, 2006.

\bibitem{Zuss}
G.~Zussman, A.~Brzezinski, and E.~Modiano, ``Multihop local pooling for
  distributed throughput maximization in wireless networks,'' in \emph{INFOCOM
  2008. The 27th Conference on Computer Communications. IEEE}.\hskip 1em plus
  0.5em minus 0.4em\relax IEEE, 2008.

\bibitem{Kim07}
J.~Kim, X.~Lin, and N.~B. Shroff, ``Locally optimized scheduling and power
  control algorithms for multi-hop wireless networks under sinr interference
  models,'' in \emph{Modeling and Optimization in Mobile, Ad Hoc and Wireless
  Networks and Workshops, 2007. WiOpt 2007. 5th International Symposium
  on}.\hskip 1em plus 0.5em minus 0.4em\relax IEEE, 2007, pp. 1--10.

\bibitem{Lee09}
H.-W. Lee, E.~Modiano, and L.~B. Le, ``Distributed throughput maximization in
  wireless networks via random power allocation,'' in \emph{Proceedings of the
  7th international conference on Modeling and Optimization in Mobile, Ad Hoc,
  and Wireless Networks}.\hskip 1em plus 0.5em minus 0.4em\relax IEEE Press,
  2009, pp. 328--336.

\bibitem{Tas98}
L.~Tassiulas, ``Linear complexity algorithms for maximum throughput in radio
  networks and input queued switches,'' in \emph{INFOCOM'98. Seventeenth Annual
  Joint Conference of the IEEE Computer and Communications Societies.
  Proceedings. IEEE}, vol.~2.\hskip 1em plus 0.5em minus 0.4em\relax Ieee,
  1998, pp. 533--539.

\bibitem{ShaBook}
D.~Shah, \emph{Gossip algorithms}.\hskip 1em plus 0.5em minus 0.4em\relax Now
  Publishers Inc, 2009.

\bibitem{LauSurvey}
Y.~Cui, V.~K. Lau, R.~Wang, H.~Huang, and S.~Zhang, ``A survey on delay-aware
  resource control for wireless systems—large deviation theory, stochastic
  lyapunov drift, and distributed stochastic learning,'' \emph{IEEE
  Transactions on Information Theory}, vol.~58, no.~3, pp. 1677--1701, 2012.

\bibitem{Chen99}
S.~Chen and K.~Nahrstedt, ``Distributed quality-of-service routing in ad hoc
  networks,'' \emph{IEEE Journal on Selected areas in Communications}, vol.~17,
  no.~8, pp. 1488--1505, 1999.

\bibitem{BocheStan}
H.~Boche and S.~Stanczak, ``Convexity of some feasible qos regions and
  asymptotic behavior of the minimum total power in cdma systems,'' \emph{IEEE
  Transactions on Communications}, vol.~52, no.~12, pp. 2190--2197, 2004.

\bibitem{Satya}
S.~Kumar, L.~Kumar, and V.~Sharma, ``Energy efficient low complexity joint
  scheduling and routing for wireless networks,'' in \emph{Modeling and
  Optimization in Mobile, Ad Hoc, and Wireless Networks (WiOpt), 2015 13th
  International Symposium on}.\hskip 1em plus 0.5em minus 0.4em\relax IEEE,
  2015, pp. 8--15.

\bibitem{GeorgBook}
L.~Georgiadis, M.~J. Neely, and L.~Tassiulas, \emph{Resource allocation and
  cross-layer control in wireless networks}.\hskip 1em plus 0.5em minus
  0.4em\relax Now Publishers Inc, 2006.

\end{thebibliography}
\bibliographystyle{IEEEtran}

\end{document}